\documentclass[sigconf,nonacm,balance=true,pbalance=true]{acmart}

\usepackage{enumitem}

\newsavebox{\takeawaybox}
\newenvironment{takeaway}{%
  \par\vskip\medskipamount\noindent
  \begin{lrbox}{\takeawaybox}\begin{minipage}{\dimexpr\columnwidth-7.5pt\relax}}%
  {\end{minipage}\end{lrbox}%
  {\color{black!45}\vrule width 1.5pt}\hspace{6pt}\usebox{\takeawaybox}\par\vskip\medskipamount}

\graphicspath{{figures/}}

\AtEndPreamble{%
  \theoremstyle{acmdefinition}%
  \newtheorem{remark}[theorem]{Remark}}

\newcommand\ainame{AI Use}%
\specialcomment{ai}{%
  \begingroup
  \section*{\ainame}
  \phantomsection\addcontentsline{toc}{section}{\ainame}
}{%
  \endgroup
  }

\newcommand{\artifacturl}{\url{https://github.com/SecretLabOU/BlockAMole}}

\newcommand{\lamburn}{\lambda_{\mathrm{burn}}}
\newcommand{\lamintro}{\lambda_{\mathrm{intro}}}
\newcommand{\lama}{\lambda_{a}}
\newcommand{\kmax}{k_{\max}}
\newcommand{\bstar}{\beta^{\star}}
\newcommand{\bstaravg}{\beta^{\star}_{\mathrm{avg}}}

\hypersetup{pdfcreationdate={D:20261006000000Z},pdfmoddate={D:20261006000000Z},
  pdfmetadate={D:20261006000000Z},pdfdate={D:20261006000000Z}}
\begin{document}

\title[Block-A-Mole]{Block-A-Mole: The Sustainability Frontier of Moving-Target Censorship Resistance}

\author{Anindya Maiti}
\orcid{0000-0001-6461-2805}
\email{am@ou.edu}
\affiliation{%
  \institution{University of Oklahoma}
  \city{Norman}
  \state{Oklahoma}
  \country{USA}
}

\author{Scott Seidenberger}
\orcid{0009-0002-6859-6042}
\email{scotty@cornell.edu}
\affiliation{%
  \institution{Cornell University}
  \city{Ithaca}
  \state{New York}
  \country{USA}
}

\begin{abstract}
Many censorship circumvention systems keep moving their servers to new cloud
addresses, so that whatever the censor discovers soon goes stale. But clients
often reach these servers through a name the censor can see and block, such as
a hostname or a front domain, and moving a server to a new address does not
give it a new name. Existing models treat each proxy as a single unit and
leave open when such a system keeps working. We model addresses and names as
two coupled layers and ask which one limits availability. Enough servers,
rotated fast enough relative to how quickly the censor finds them one at a time,
keep some address reachable, and we compute how fast is fast enough. Past that point,
faster rotation buys almost nothing, and the names become the limit. If the
censor burns names, blocking them for good, $\beta$ times as fast as the system
mints new ones, clients can connect less than a $1/\beta$ share of the time,
however fast the servers move. For an illustrative system with eight servers
and room for eight names, sessions lasting five mean mint intervals
stay connected with probability at least $0.95$ only
while $\beta\le0.66$, the system's sustainability frontier, which a censor that
times its blocks can cut about tenfold. Public data show that China and Iran block platform names one tenant at a
time or a whole suffix at once, that Russia orders such blocks, and that even
Snowflake clients need names to obtain proxies. Our
solver, simulator and data pipelines are openly available. Once rotation is
fast enough, the race that matters is between minting and burning names, not
between rotation and discovery.
\end{abstract}


\maketitle

\section{Introduction}\label{sec:intro}

Static circumvention endpoints can be enumerated. Tor bridges have been
harvested at scale from their distribution channels~\cite{bridge_discovery},
and China's Great Firewall (GFW) blocked new Tor Browser bridges 2 to 36 days
after release~\cite{censors_delay}. Moving-target systems instead move
endpoints across commodity clouds so that the censor's discovery goes stale.
SpotProxy~\cite{spotproxy} keeps respawning proxies at new IP addresses.
NinjaDoH~\cite{ninjadoh} rotates a DNS-over-HTTPS (DoH) resolver's addresses
without domain names. CensorLess~\cite{censorless} regularly replaces its
serverless-function bridges to avoid DNS and TLS Server Name Indication (SNI)
blocking. These designs rely on collateral freedom~\cite{collateral_freedom},
since blocking a cloud's address ranges would also cut off its legitimate
tenants.

Many deployments also depend on names, which censors have long blocked in DNS,
HTTP and TLS traffic~\cite{gfwatch,gfweb}, and Russia and then China extended
this to QUIC~\cite{russia_quic_sni,gfw_quic}. We call a deployment
\emph{name-dependent} if a client cannot reach an endpoint without
presenting a name the censor can observe and block persistently, such as a
CensorLess function's hostname or a front domain for the rendezvous of
Snowflake~\cite{snowflake} (Section~\ref{sec:threat}). For such a deployment,
rotation renews the address but not the name.

Proxy distribution when some users are censor agents has been studied with
algorithms and game
theory~\cite{mahdian_fighting,torbricks,nasr_enemy_gateways,fares_game_changed},
and simulations find that the rate of new proxies is key to
availability~\cite{nasr_enemy_gateways,fares_game_changed}. These models treat
each proxy as one unit. Three questions therefore remain open for
name-dependent deployments, namely when rotation wins, which resource binds,
and whether the right lever is faster rotation or replacing burned names fast enough.

To answer them, we model such a deployment as two layers coupled through
reachability (Section~\ref{sec:model}). On the address layer, the defender
rotates endpoints to fresh addresses while the censor finds and blocks them, which per endpoint is FlipIt~\cite{flipit} with
exponential strategies. On
the name layer, the defender keeps a buffer of up to $\kmax$ live names and
mints new ones at rate $\lamintro$, and the censor \emph{burns} live names,
that is, blocks them persistently, at an effective rate $\lamburn$ that
discounts what it fails to see. A client can connect while some endpoint has
an unblocked address and some name is live. The \emph{burn ratio} is
$\beta=\lamburn/\lamintro$, and time is measured in mean mint intervals
$1/\lamintro$.

Redundancy and rotation fast enough relative to discovery win the address
layer, provided the censor discovers endpoints independently. The required
speed, which we compute, rises with session length measured in mean
address-discovery delays (Section~\ref{sec:rotation}). The name layer binds,
since the fraction of time a client can connect stays below $1/\beta$ however
fast endpoints rotate (Theorem~\ref{thm:cap}). Hence no configuration keeps
sessions of length $T$ connected throughout with probability at least $\alpha$ once
$\beta\ge1/\alpha$ (Definition~\ref{def:avail}).

The largest burn ratio a deployment can sustain while meeting this target is
its \emph{sustainability frontier} $\bstar(\alpha,T)$. In an illustrative
configuration with eight endpoints, each with a mean discovery delay of one mint
interval and rotated three times per mint interval on average, $\bstar(0.95,5)$ is $0.66$ for a buffer of $\kmax=8$ names and
exceeds~1 only for $\kmax\ge63$ (Section~\ref{sec:frontier}). These values
hold against a censor that burns names at a constant rate. A censor that
strikes only after finding every name in use, at the same discovery rate,
lowers $\bstar(0.95,5)$ about tenfold, to $0.064$, and a small hidden reserve of
names recovers part of the loss (Section~\ref{sec:ext}). Names are
cheap to mint~\cite{porkbun_pricing}, but if the censor burns names on
exposure, for example in public Certificate Transparency (CT) logs, minting
faster helps only with names that stay unseen until use.

Public data ground these ingredients (Section~\ref{sec:data}). The unit a
censor burns varies by censor and platform. China has blocked \texttt{workers.dev}
suffix-wide and Iran has likewise blocked \texttt{pages.dev}, yet both block \texttt{github.io}
tenants one name at a
time~\cite{gfwatch,ooni_data,irblock}. When Fastly ended support in March 2024
for the domain
fronting used for Snowflake's rendezvous~\cite{tor_metrics_timeline}, estimated
Snowflake users worldwide fell
32.5\% within a week while the number of proxy addresses barely
changed~\cite{tor_collector,tor_metrics}.

Our contributions are as follows.
\begin{itemize}[leftmargin=*]
\item A two-layer timing model extending FlipIt to a second, burnable
  resource (Sections~\ref{sec:threat} and~\ref{sec:model}).
\item The separation of the layers, with a rotation speed that suffices for the
  address layer and a $1/\beta$ cap on the name layer at any speed
  (Theorem~\ref{thm:cap} and Section~\ref{sec:rotation}).
\item Closed-form availability, the cap, the exact frontier, and extensions
  for exposure at mint, strategic block timing and correlated takedowns
  (Section~\ref{sec:analysis}).
\item Grounding of the model in public data on which names censors burn and
  on Snowflake's rendezvous names (Section~\ref{sec:data}).
\item An operating recipe that sizes rotation and sets the mint rate to keep
  $\beta$ below $\bstar$, with a consistent estimator of $\beta$
  (Section~\ref{sec:implications}).
\item An exact solver and an event-driven simulator, available with our data
  pipelines as an open testbed for variants beyond the closed forms
  (Appendix~\ref{app:sim}).
\end{itemize}

Beyond the required speed, the operative control of a name-dependent
deployment is the rate at which fresh names are minted relative to the rate at
which the censor burns them, not the rotation interval. Operators can measure
$\beta$ and keep it below $\bstar$.

\section{Background and Related Work}\label{sec:related}

\paragraph{Collateral freedom and its limits}
Robinson, Yu and An coined the term \emph{collateral freedom} for circumvention
that works because blocking it would cost the censor too
much~\cite{collateral_freedom}. In their leading example, the GFW poisoned DNS
for the Google App Engine names under \texttt{appspot.com}
but left the Google addresses that served GoAgent unblocked except for one
day, because the same addresses served Google services that businesses
needed~\cite{collateral_freedom}. Collateral protected the shared addresses,
not the names on them, so name dependence is often the price of collateral
protection. Domain fronting applies collateral freedom to names, since
blocking a fronted connection means blocking the front domain with all its
legitimate traffic~\cite{domain_fronting}. Censors also accept large
collateral at times. Iran has blocked entire top-level domains such as
\texttt{.il}~\cite{irblock}, and when Russia moved to block Telegram in April
2018, it blocked millions of cloud addresses and made unrelated services on
Google and Amazon unreachable~\cite{russia_decentralized}.

\paragraph{Timing games and moving-target defense}
Moving-target defense (MTD) changes a system's configuration over time to deny
attackers a static target~\cite{mtd_survey_sengupta,mtd_survey_cho}, for
example by mutating hosts' IP addresses so that learned addresses go
stale~\cite{mt6d,ofrhm,nasr_randomization}. FlipIt models such contests as a
timing game in which either player seizes a shared resource by moving and
learns who held it only by moving~\cite{flipit}, and Li and Zheng model MTD as
a leader--follower game over when and where to migrate~\cite{stackelberg_mtd}.
Per endpoint, our address layer is FlipIt with exponential strategies and no
move costs (Lemma~\ref{lem:addr}). FlipThem extends FlipIt to several
resources, of which the attacker must hold all (AND) or any one (OR) to control
the system~\cite{flipthem}. Our address layer, which fails only when all $n$
endpoints are blocked, is the AND case, whether the censor discovers endpoints
independently (Lemma~\ref{lem:addr}) or all at once (Remark~\ref{rem:sync}).
In these games the defender wins every resource back by moving. The name is a
second resource, which no move wins back once burned. Only minting replenishes
the pool of live names, so availability depends on a birth--death process for
that pool, coupled to the address layer through reachability.

\paragraph{Game-theoretic proxy distribution}
Nasr et al.\ model proxy assignment as a game in which the censor decides at
each step whether to block a discovered proxy or leave it open to identify its
users, and in their simulations a system is eventually blocked below a critical
arrival rate of new proxies~\cite{nasr_enemy_gateways}. SpotProxy and
CensorLess are evaluated in this framework~\cite{spotproxy,censorless}, and
CensorLess's simulation pairs each bridge URL with a single
address~\cite{censorless}. Fares et al.\ extend the framework to ephemeral
proxies such as Snowflake's, to censors that find proxies by traffic analysis
and to multiple censors, and find that the arrival rate of proxies relative to
users matters more than the censor's strategy. They identify blocking-resistant
control channels, which deliver proxy assignments to users, as a bottleneck
their simulator does not capture~\cite{fares_game_changed}. Snowflake's
rendezvous is such a channel, and it depends on names
(Section~\ref{sec:snowflake}).

These models treat a proxy as one unit that the censor blocks as a whole. Ours
complements them by taking discovery as an effective rate and separating the
address, which rotation renews, from the name, which it does not. This
separation yields a cap of $1/\beta$ on time-average availability that no
rotation speed escapes (Theorem~\ref{thm:cap}), and we obtain availability in
closed form and the frontier $\bstar$ exactly rather than by simulation. The
resulting minting condition ${\lamintro\ge\lamburn/\bstar}$
(Section~\ref{sec:implications}) is the name-layer counterpart of Nasr et al.'s
critical arrival rate, tied to the burn rate rather than to client arrivals.

\section{Threat Model and Scope}\label{sec:threat}

We consider a national censor with the capabilities documented for China's
GFW, Iran's national filter and Russia's TSPU, the deep packet inspection (DPI)
devices that Russia's regulator distributes to Internet service providers
(ISPs)~\cite{gfw_hidden,gfw_quic,tspu,iran_whitelister,irblock}. The
defender's moving-target deployment rotates cloud endpoints between addresses
and, if clients reach them by name, replaces blocked names from a pool.

\paragraph{Position and capabilities}
The censor sits at international gateways and inside ISP networks, where it can
observe, drop and inject packets. Its DPI reads names from DNS queries, the
HTTP Host header and the TLS SNI~\cite{gfwatch,gfweb,wall_behind_wall}. The
TSPU switched to SNI-based censorship of QUIC between May 2022 and July
2023~\cite{russia_quic_sni}, and the GFW has filtered QUIC by SNI since 7~April
2024~\cite{gfw_quic}. The censor finds endpoints by classifying
traffic~\cite{gfw_fullyencrypted}, probing suspected
servers~\cite{gfw_hidden,unpublished_bridges} and harvesting distribution
channels~\cite{bridge_discovery,censors_delay}, and it may learn a name before
any client uses it, for example from CT logs (Section~\ref{sec:params}). It can
block a discovered endpoint's address, its name or both.

\paragraph{Blocks persist}
A burned name stays burned. The GFW's national name blocks last a median of at
least 256 days, and lapses, as at Henan's regional firewall, only help the
defender~\cite{wall_behind_wall}. An address block need only last while the
endpoint keeps the address. The GFW renews it during use and lifts it about 12
hours after the service stops~\cite{unpublished_bridges}, so persistence is
conservative for the GFW. Rotations are assumed to land on unblocked addresses,
which nearly holds while blocked addresses are a small share of the provider's space.

\paragraph{Discovery is an effective rate}
The censor need not see every connection or block everything it sees. The
model summarizes discovery by two effective rates, $\lama$ for each endpoint's
address and $\lamburn$ for the defender's names (Section~\ref{sec:model}). Each
is a rate of attempts times the probability that one finds and blocks its
target. Attempts include probes, traffic classification, channel enumeration
and agents posing as users~\cite{nasr_enemy_gateways}. Both factors
are limited in practice. The GFW applies its filter for fully encrypted traffic
only to about a quarter of the connections toward certain data-center
ranges~\cite{gfw_fullyencrypted}. All model results are functions of
these rates.

\paragraph{Collateral and provider actions}
The collateral damage a censor accepts differs between normal operation and
crisis, and collateral protects shared suffixes and front domains only
partially (Section~\ref{sec:data}). The model lets the censor block any
discovered address or name at no cost, which is conservative for the defender,
and collateral enters through the granularity of burns
(Section~\ref{sec:burn}). Blocking whole address ranges in a crisis, as Russia
did in 2018 (Section~\ref{sec:related}), is a regime change outside the model.
Providers can also remove many names at once, through account suspensions,
abuse takedowns or policy changes. We model such events as correlated burns and
study spreading names over providers (Section~\ref{sec:ext}). The model omits
compute suspensions, which remove endpoints and names together.

\paragraph{Scope}
The analysis applies to \emph{name-dependent} deployments, whose clients must
present a name the censor can observe and block persistently. They include
serverless and platform tenants addressed by hostname, as in
CensorLess~\cite{censorless}, and proxies behind a content delivery network
(CDN), reached through a tenant hostname or a front
domain~\cite{domain_fronting}. WebTunnel bridges by default use the operator's
domain name with a publicly trusted certificate~\cite{webtunnel}, and stock DoH
clients reach a resolver by the hostname in its URL. \emph{Nameless} designs
present no such name. Examples are REALITY, which presents an unrelated site's
name~\cite{reality}, peer-to-peer data paths, and a pinned certificate with no
SNI or, as WebTunnel supports, an arbitrary
one~\cite{tor_blog_2025,tor_gitlab_40064}. Nameless designs have no name
layer, and their availability is set by the address layer
(Lemma~\ref{lem:addr}, under independent discovery) and by protocol
fingerprinting. NinjaDoH is nameless, since clients reach its resolver by
address and learn new addresses through the InterPlanetary Name System (IPNS)
rather than DNS~\cite{ninjadoh}.
CensorLess is name-dependent. SpotProxy and Snowflake are nameless on their
data paths but depend on names to obtain proxies. A SpotProxy client joins, and
must \mbox{reregister} whenever its proxy is blocked, through a distribution service
such as a domain-fronted broker~\cite{spotproxy}, and a Snowflake client needs
its rendezvous for every new proxy (Section~\ref{sec:snowflake}).

\paragraph{Out of scope}
The model excludes protocol fingerprinting, which lets a censor block a
transport whatever its address or name, as Russia did with Snowflake's Datagram
TLS (DTLS) handshake~\cite{snowflake}. Residual blocks, lasting seconds to
minutes after a detected
connection~\cite{gfw_fullyencrypted,gfw_quic,tspu,iran_whitelister}, are
neither address nor name events. We assume that clients learn the current
endpoints and names through a bootstrap channel, and we exclude national
shutdowns, which cut off every provider at once.

\section{Model}\label{sec:model}

We model a name-dependent deployment as two resources that the censor depletes
and the defender replenishes by rotating endpoint addresses and minting names.
A \emph{name} is the unit that the censor blocks in one action. Depending on
the censor and the platform, it is a registrable domain with all its
subdomains, a tenant hostname on a shared platform, such as
\texttt{<tenant>.github.io}, or an exact hostname (Section~\ref{sec:burn}).
\emph{Minting} makes a new name usable, and a name is \emph{burned} when the
censor blocks it persistently. Rotation changes an endpoint's address, not its
name. Unless stated otherwise, time is measured in mean mint intervals
$1/\lamintro$, so a window length $T$ counts mint intervals.

\begin{definition}[Two-layer moving-target model]\label{def:game}
A configuration consists of the defender's choices $n$, $\mu$, $\kmax$ and
$\lamintro$ and the censor's rates $\lama$ and $\lamburn$, all positive. Events
occur at the points of independent Poisson processes with the stated rates.
\begin{enumerate}[leftmargin=*]
\item \emph{Address layer.} The defender runs $n$ endpoints. Each endpoint
rotates at rate $\mu$ to a fresh address that the censor has not blocked. The
censor discovers and blocks each unblocked endpoint at rate $\lama$. A blocked
endpoint stays blocked until its next rotation. Discovery and rotation are
independent across endpoints.
\item \emph{Name layer.} The defender keeps a pool of
$K(t)\in\{0,\dots,\kmax\}$ live (unblocked) names. It mints a new name at rate
$\lamintro$ while $K<\kmax$. The censor burns one live name at rate $\lamburn$
while $K\ge1$. Endpoints use any live name and move to another one when theirs
is burned.
\item \emph{Censor class.} The rates $\lama$ and $\lamburn$ are constant, so
the censor does not time its blocks to the defender's state. The name process
is independent of the address processes.
\item \emph{Reachability.} The deployment is reachable at time $t$ if
$K(t)\ge1$ and at least one endpoint is unblocked.
\end{enumerate}
\end{definition}

The rates are the effective rates of Section~\ref{sec:threat}. The burn rate
$\lamburn$ counts persistent blocks, not sightings, and is an aggregate rate
that does not grow with $K$ (Appendix~\ref{app:ext}
relaxes this). The product form of Theorem~\ref{thm:closed} assumes that any
live name can
reach any unblocked endpoint, as when every endpoint accepts every live name,
and that the layers are independent, although one observation can reveal an
endpoint's address and name together. Pairing names with particular endpoints
never raises availability and leaves the cap of Theorem~\ref{thm:cap} intact.
Section~\ref{sec:ext} goes beyond the constant-rate censor to treat names
exposed at mint, blocks timed to the defender's state and burns in bursts.

\begin{definition}[Availability]\label{def:avail}
Fix a configuration and let the model of Definition~\ref{def:game} run in
stationarity. The \emph{time-average availability} $A$ is the probability
that the deployment is reachable at a given instant. For $T>0$, the
\emph{interval availability} $A_T$ is the probability that the deployment is
reachable at every instant of a window $[t,t+T]$. By stationarity it does not
depend on $t$. The configuration is \emph{$(\alpha,T)$-available} if
$A_T\ge\alpha$.
\end{definition}

The time-average availability $A$ is also the long-run fraction of time the
deployment is reachable, and $A_T\le A$ for every $T$. The interval metric
counts any outage in the window as a failure, however brief. It describes a
session that survives a change of endpoint or name, which requires a transport
that can move sessions, such as Turbo Tunnel~\cite{turbotunnel}.

\paragraph{Address layer}
Each endpoint alternates between an unblocked and a blocked state.

\begin{lemma}[Address layer]\label{lem:addr}
Under Definition~\ref{def:game}, each endpoint is blocked with stationary
probability
\[
  p = \frac{\lama}{\lama+\mu}.
\]
Because discovery and rotation are independent across endpoints, all $n$ are
blocked with probability $p^{n}$. The \emph{address factor}, the stationary
probability that at least one endpoint is unblocked, is $c_a=1-p^{n}$.
\end{lemma}

Appendix~\ref{app:proofs} proves this lemma and Proposition~\ref{prop:pool}
below. The lemma restates FlipIt's Theorem~2 and FlipThem's AND model
for exponential strategies~\cite{flipit,flipthem} and describes one instant. Two
assumptions behind the lemma need care. Under
Poisson rotation and an independent discovery delay $D$ with any law, an
endpoint is blocked with probability $p=\mathbb{E}[e^{-\mu D}]$, which neither
the mean nor the tail of $D$ fixes. The exponential model is not conservative,
so an operator should size $n$ from the measured blocked fraction $p$
(Remark~\ref{rem:delay} in Appendix~\ref{app:exact}). Discovery must also be
independent across endpoints.

\begin{remark}[Synchronized discovery]\label{rem:sync}
The geometric decay of $p^{n}$ needs discovery, not only rotation, to be
independent across endpoints. Suppose instead that discovery comes in sweeps
that each block every endpoint at once, as when the censor reads a directory
listing them all. With sweeps at rate $\lama$, each endpoint is still
blocked with probability $p$, but an outage lasts until any endpoint rotates,
so all $n$ are blocked with probability $\lama/(\lama+n\mu)$~\cite{flipthem},
which decays like $1/n$ rather than geometrically. At $\mu=3\lama$ and $n=8$
this is $0.040$, against $p^{n}=1.5\times10^{-5}$ under independent discovery.
If also $\lama=\lamintro$, the address layer stays up throughout a window of
$T=5$ with probability $0.998$ under independent discovery but only $0.949$
when sweeps make up $1\%$ of each endpoint's discovery rate. With a sweep share
$s$, it stays up throughout a window of length $T$ with probability at most
$e^{-s\lama T}$, whatever $n$ and $\mu$, because one sweep inside the window
causes an outage. Independence must therefore be engineered. Sweeps arise from
any channel that shows the censor all current endpoints, such as a
public directory, the DNS
records of a name that resolves to them or a list sent to every client, which
a censor posing as a user also receives.
\end{remark}

\begin{takeaway}
In summary, extra endpoints help a lot only if the censor has to find them one
at a time.
\end{takeaway}

\paragraph{Name layer}
We call
\begin{equation}\label{eq:beta}
  \beta = \frac{\lamburn}{\lamintro}
\end{equation}
the \emph{burn ratio}, the mean number of names that the censor burns per mint
interval while the pool is nonempty.

\begin{proposition}[Name pool]\label{prop:pool}
Under Definition~\ref{def:game}, $K(t)$ is a birth--death chain on
$\{0,\dots,\kmax\}$ with birth rate $\lamintro$ in states $k<\kmax$ and death
rate $\lamburn$ in states $k\ge1$. Its stationary law is
$\pi_k\propto\beta^{-k}$, and the pool is empty with probability
\begin{equation}\label{eq:pi0}
  \pi_0 =
  \begin{cases}
    \dfrac{1-\beta^{-1}}{1-\beta^{-(\kmax+1)}}, & \beta\neq1,\\[2.2ex]
    \dfrac{1}{\kmax+1}, & \beta=1.
  \end{cases}
\end{equation}
In each state $0<k<\kmax$ the mean drift of\/ $K$ is $\lamintro(1-\beta)$. The
pool fills when $\beta<1$, and then $\pi_0\le\beta^{\kmax}$. It drains when
$\beta>1$.
\end{proposition}

\begin{takeaway}
In summary, when names are minted faster than they are burned ($\beta<1$), the
chance that no name is live shrinks geometrically as the buffer grows.
\end{takeaway}

\section{Analysis}\label{sec:analysis}

Unless stated otherwise, numbers in this section refer to the canonical
configuration, with $n=8$ endpoints, $\lama=\lamintro$, $\mu=3\lama$, $\kmax=8$,
$\alpha=0.95$ and $T=5$. This illustrative configuration assumes that addresses
are discovered on the
timescale of minting, and Section~\ref{sec:rotation} examines faster discovery.

\subsection{Closed-Form Availability and the Burn Cap}\label{sec:closed}

\begin{theorem}[Closed-form availability]\label{thm:closed}
Under Definition~\ref{def:game}, with $p$ and the address factor $c_a=1-p^n$ as in
Lemma~\ref{lem:addr},
\begin{equation}\label{eq:closed}
  A=\bigl(1-\pi_0(\beta,\kmax)\bigr)\,c_a
   =\bigl(1-\pi_0(\beta,\kmax)\bigr)\bigl(1-p^{n}\bigr),
\end{equation}
where $\pi_0(\beta,\kmax)$ is the stationary probability that no live name exists
(Proposition~\ref{prop:pool}).
\end{theorem}

Appendix~\ref{app:proofs} proves it through the joint chain of the two layers.
For a nameless deployment, $A=c_a$. The cap below is an impossibility result.
Its bound $1/\beta$ needs neither
independent layers nor Poisson events (Proposition~\ref{prop:conserve}).

\begin{theorem}[Burn cap]\label{thm:cap}
Let the name pool follow Proposition~\ref{prop:pool}. For every $\beta>0$ and
every $n$, $\mu$, $\lama$ and $\kmax$,
\[
  A_T\le A\le1-\pi_0(\beta,\kmax)<\frac{1}{\beta},
\]
and $1-\pi_0(\beta,\kmax)$ increases to $1/\beta$ as
$\kmax\to\infty$ if\/ $\beta>1$. Hence no configuration with $\beta\ge1/\alpha$ is
$(\alpha,T)$-available, for any $T>0$, whatever its rotation speed and number of
endpoints. If the two layers are independent, as in Definition~\ref{def:game},
then also $A=c_a(1-\pi_0)<c_a/\beta$, so the same holds for every
$\beta\ge c_a/\alpha$. For $\beta>1$ the cap is tight, since $A\to1/\beta$ as
$\kmax\to\infty$ and $c_a\to1$.
\end{theorem}

\begin{proof}
A window that is reachable throughout is reachable at its start, so $A_T\le A$.
A reachable deployment needs a live name, so $A\le\Pr[K\ge1]=1-\pi_0$, whatever
the address layer does. Lemma~\ref{lem:pi0}(c) in Appendix~\ref{app:proofs}
gives $1-\pi_0<1/\beta$, and Lemma~\ref{lem:pi0}(d) the limit. If
$\beta\ge1/\alpha$, then $A_T<1/\beta\le\alpha$. No step so far uses
independence. With independent layers, Theorem~\ref{thm:closed} gives
$A=c_a(1-\pi_0)<c_a/\beta$, and the limits of $c_a$ and $1-\pi_0$ give
tightness.
\end{proof}

\begin{takeaway}
In summary, if names burn twice as fast as they are minted, clients can connect
under half the time, whatever the endpoints do.
\end{takeaway}

\paragraph{What binding means}
Theorem~\ref{thm:cap} is the sense in which the name layer binds. Rotation speed
enters~\eqref{eq:closed} only through $c_a$, and under the interval metric also
through the address window (Proposition~\ref{prop:exact}). Once both are close to~$1$,
the name pool alone decides whether the target is met, however fast the endpoints
rotate. Binding does not mean that names are expensive. Cheap minting
(Section~\ref{sec:params}) keeps $\beta$ small but is no way around the cap, and a
stock of names bought in advance amounts to minting at its release rate or to a
larger
buffer, which never lifts $\bstar$ beyond $c_a/\alpha$ (Theorem~\ref{thm:frontier}).

\subsection{The Sustainability Frontier}\label{sec:frontier}

For a target $\alpha\in(0,1)$, the time-average frontier $\bstaravg$ is the largest
$\beta$ with $A\ge\alpha$, and the interval frontier $\bstar=\bstar(\alpha,T)$ is
the largest $\beta$ with $A_T\ge\alpha$. Each is $0$ if no $\beta>0$ meets the
target. We call $\bstar(\alpha,T)$ the \emph{sustainability frontier}, the largest
burn ratio that a deployment can sustain while meeting its availability target
(Figure~\ref{fig:frontier}). On common clocks a higher burn rate never leaves more
live names (Lemma~\ref{lem:coupling}), so $A$ and $A_T$ are nonincreasing in
$\beta$. A configuration therefore meets its target exactly when $\beta$ does
not exceed the frontier.

\begin{theorem}[Frontiers]\label{thm:frontier}
Fix $\alpha\in(0,1)$, $n$, $\mu$, $\lama>0$ and $\kmax$, under
Definition~\ref{def:game}.
\begin{enumerate}[label=\textup{(\roman*)},leftmargin=*]
\item If\/ $c_a>\alpha$, then $\bstaravg$ is the unique $\beta>0$ with
  $1-\pi_0(\beta,\kmax)=\alpha/c_a$. If\/ $c_a\le\alpha$, no $\beta>0$ gives
  $A\ge\alpha$.
\item If\/ $c_a>\alpha$, then $\bstaravg$ is strictly increasing in $\kmax$, $n$ and
  $\mu$.
\item $\bstaravg<c_a/\alpha$, and if\/ $c_a>\alpha$, then $\bstaravg\to c_a/\alpha$ as
  $\kmax\to\infty$.
\item $\bstar(\alpha,T)$ is nondecreasing in $\kmax$, $n$ and $\mu$, and
  \[
    \bstar(\alpha,T)\le\bstaravg<c_a/\alpha\le1/\alpha .
  \]
\end{enumerate}
\end{theorem}

Appendix~\ref{app:proofs} gives the proof. Without independence of the layers,
$\bstar\le\bstaravg<1/\alpha$ still holds by Theorem~\ref{thm:cap}.

\begin{proposition}[Exact interval availability]\label{prop:exact}
Under Definition~\ref{def:game},
\begin{equation}\label{eq:exact}
  A_T=W^{\mathrm{addr}}_T\,W^{\mathrm{name}}_T,
  \qquad
  W_T=\boldsymbol{\pi}_{\mathcal U}\,e^{Q_{\mathcal U\mathcal U}T}\,\mathbf 1 .
\end{equation}
Here the \emph{address window} $W^{\mathrm{addr}}_T$ is the probability that the
stationary number $c$ of unblocked endpoints stays up ($c\ge1$) throughout a
window of length $T$, and the \emph{name window} $W^{\mathrm{name}}_T$ is the same
for the stationary name pool $K$ (up when $K\ge1$). For a finite
chain with
generator $Q$, stationary law $\boldsymbol{\pi}$ and up set $\mathcal U$, $\boldsymbol{\pi}_{\mathcal U}$
and $Q_{\mathcal U\mathcal U}$ are $\boldsymbol{\pi}$ and $Q$ restricted to $\mathcal U$, and
$\mathbf 1$ is the all-ones vector.
\end{proposition}

Appendix~\ref{app:exact} gives the proof. The address chain moves from $c$ to
$c+1$ at rate $\mu(n-c)$ and to $c-1$ at rate $\lama c$. Both factors thus need
exponentials of matrices with at most $\max\{n,\kmax\}$ rows,
computed to about $10^{-12}$, and a root $\bstar$ takes milliseconds. Every
frontier in this paper is an exact value of the model, not an
estimate from noisy simulation runs. An independent event-driven simulator
of the same Markov model agrees with the exact values within Monte Carlo error
where standard errors are available, and to within $0.009$ elsewhere
(Appendix~\ref{app:exact}). Both assume independent layers, so their agreement
checks the code, not the model.

\paragraph{Exact frontiers}
Both frontiers rise with the buffer (Figure~\ref{fig:frontier}). At
$\kmax=8$ the time-average frontier is $\bstaravg=0.840$ and the interval
frontier is lower, $\bstar(0.95,5)=0.661$. The address window is
$W^{\mathrm{addr}}_5=0.998$
here, so the name pool sets $\bstar$. The interval frontier exceeds~$1$ from
$\kmax=63$ on and tends to $1.016$. Longer windows lower
the frontier gently, so the smallest buffer with $\bstar>1$ grows with the window
length, as it does with the target.

\begin{figure}[t]
\centering
\includegraphics[width=\columnwidth]{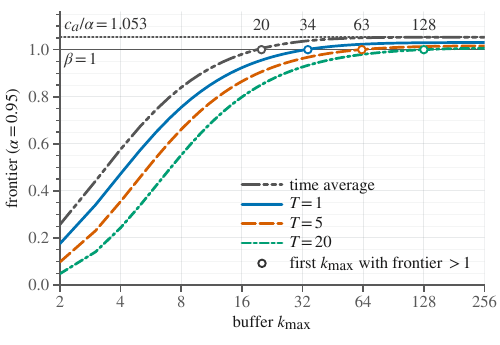}
\caption{The time-average frontier $\bstaravg$ and the interval frontier
$\bstar(0.95,T)$ for $T=1$, $5$ and $20$ against the buffer $\kmax$ (log scale),
at $\alpha=0.95$ ($n=8$, $\mu/\lama=3$, $\lama=\lamintro$). Circles mark the
smallest buffer at which each frontier exceeds $1$ ($\kmax=20$, $34$, $63$ and
$128$). The dotted line is the bound $c_a/\alpha=1.053$.}
\Description{Line plot of the frontier for a 0.95 target, on a vertical axis
from 0 to about 1.15, against the buffer size on a logarithmic axis from 2 to
256 names. Four curves show the time-average frontier and the interval
frontier for session lengths 1, 5 and 20. All rise steeply from between 0.05
and 0.26 at a buffer of 2 and level off just above 1, with longer sessions
lower. Circles on the horizontal line at 1 mark where each curve exceeds 1, at
buffers of 20 for the time average and 34, 63 and 128 for session lengths 1, 5
and 20. A dotted line at 1.053 marks the upper bound.}
\label{fig:frontier}
\end{figure}

\begin{figure}[t]
\centering
\includegraphics[width=\columnwidth]{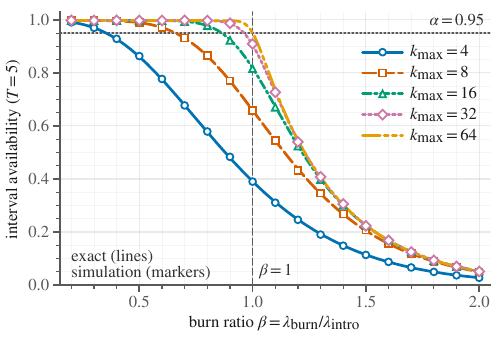}
\caption{Interval availability $A_T$ ($T=5$) against the burn ratio $\beta$
for buffers $\kmax\in\{4,8,16,32,64\}$ ($n=8$, $\mu/\lama=3$,
$\lama=\lamintro$). Lines are exact values and markers event-driven
simulations ($\kmax\le32$), with two-standard-error bars smaller than the
markers. The horizontal dotted line marks the target $\alpha=0.95$ and
the vertical dashed line $\beta=1$.}
\Description{Line plot of interval availability, from 0 to 1, against the
burn ratio, from 0.2 to 2, for buffers of 4, 8, 16, 32 and 64 names. Each
curve starts near 1 and falls as the burn ratio grows. Larger buffers stay
near 1 longer and then fall more steeply around a burn ratio of 1, marked by a
dashed vertical line. The buffer-8 curve crosses the dotted 0.95 target line
near a burn ratio of 0.66 and the buffer-64 curve near 1. At a burn ratio of 2
all curves are at about 0.05 or below. Simulation markers lie on the exact
lines.}
\label{fig:phase}
\end{figure}

For finite buffers $A_T$ falls continuously as $\beta$ grows
(Figure~\ref{fig:phase}), so $\bstar(\alpha,T)$ is a threshold for a chosen target
and window, not a phase transition. Only as $\kmax\to\infty$ does $\beta=1$ become a
sharp point, where the limit $c_a\min\{1,1/\beta\}$ of the time-average availability
has a kink. The interval metric is much stricter, since at $\kmax=8$ and $\beta=1$
the time-average availability is $0.889$ but $A_5=0.659$.

\begin{takeaway}
In summary, a larger buffer raises the frontier, but never above $1/\alpha$, so
no buffer withstands burning much faster than minting.
\end{takeaway}

\subsection{When Rotation Speed Matters}\label{sec:rotation}

We now relax $\lama=\lamintro$. Under the time-average metric the address layer
enters only through $c_a=1-p^n$ (Lemma~\ref{lem:addr}), which depends on $n$ and
$\mu/\lama$ alone. The address window $W^{\mathrm{addr}}_T$ also depends on
$\lama T$, the session length in mean discovery delays, because a session fails
if every endpoint is blocked at any moment within it. At $n=8$ and
$\mu/\lama=1$, $c_a=0.996$, yet $W^{\mathrm{addr}}_T$ is $0.872$ for $\lama T=5$,
$0.669$ for $\lama T=15$ and $0.204$ for $\lama T=60$.

Rotation speed therefore matters up to a threshold that rises with $\lama T$. The
interval frontier $\bstar(0.95,5)$ stays within $0.01$ of its value for a perfect
address layer ($0.665$) once $\mu/\lama\ge2.52$, $4.00$ and $5.66$ for
$\lama T=5$, $50$ and $300$ ($\lama=\lamintro$, $10\lamintro$ and $60\lamintro$).
For the time-average frontier the threshold is about $1.12$ for any $\lama$.
These thresholds are grid values, at most $7\%$ above the true ones. Once
rotation reaches the threshold, the name economy sets the frontier. Faster
rotation then moves it by less than $0.01$, and the deployment meets its target
by minting fast enough to keep $\beta$ below the frontier. Below the threshold,
slower rotation lowers the frontier, and once $W^{\mathrm{addr}}_T$ falls below
$\alpha$ no $\beta$ meets the target. At the canonical $\mu/\lama=3$ this
happens once $\lama>29.5\lamintro$.

\paragraph{Wall-clock timescales}
Published censor behavior is reported in wall-clock time. With $D$ the time
from putting an address in use to its block and $\tau$ the wall-clock mean mint
interval, exponential discovery gives $\lama=\tau/\mathbb{E}[D]$, and a session
of wall-clock length $T_{\mathrm{wall}}$ has
$\lama T=T_{\mathrm{wall}}/\mathbb{E}[D]$, independent of $\tau$.
For Tor bridges that the GFW confirms by active probing, blocks follow on the
order of minutes, although no source reports a
mean~\cite{gfw_blocking_tor,gfw_hidden,unpublished_bridges}, and relays listed
in the public Tor consensus were blocked after about 10~minutes~\cite{unpublished_bridges}. Taking a mean time to block of one minute
as an order of magnitude, a one-hour session has $\lama T=60$. With $n=8$
endpoints, $W^{\mathrm{addr}}_T$ then reaches $0.95$ only if $\mu/\lama\ge2.49$,
a rotation about every $24$\,s per endpoint. A one-day session
($\lama T=1{,}440$) needs $\mu/\lama\ge4.63$, a rotation about every $13$\,s. Both
are stress
tests that provider quotas may rule out. These bare minimums leave the name
layer no margin but suffice for nameless
designs. They assume independent, exponential discovery. Delays that are
often much shorter than their mean require more, and sweeps must be avoided by
engineering independent discovery (Remarks~\ref{rem:sync} and~\ref{rem:delay}).

\begin{takeaway}
In summary, with independent discovery, rotation only has to be fast enough, and how fast is enough depends on
session length relative to how quickly the censor finds an address.
\end{takeaway}

\subsection{Beyond the Constant-Rate Censor}\label{sec:ext}

We now leave Definition~\ref{def:game} one assumption at a time
(Appendix~\ref{app:ext}). Unless stated otherwise, the numbers are exact
frontiers $\bstar(0.95,5)$ at the canonical configuration, where the
constant-rate frontier is $0.661$.

\paragraph{Mint-time exposure}
A name can become visible at mint, before any client uses it, for example when its
certificate is logged in CT (Section~\ref{sec:params}). Let each mint be exposed
independently with probability $\rho$, and let the censor burn exposed names
preemptively at rate $\theta$ per name, on top of usage burns at rate $\lamburn$.

\begin{proposition}[Mint-time exposure]\label{prop:exposure}
If exposed names are burned at once ($\theta\to\infty$), the usable names form the
pool of Proposition~\ref{prop:pool} with mint rate $(1-\rho)\lamintro$.
Theorems~\ref{thm:closed} to~\ref{thm:frontier} and Proposition~\ref{prop:exact}
then hold with $\beta$ replaced by the effective burn ratio
$\beta_{\mathrm{eff}}=\lamburn/((1-\rho)\lamintro)$ and time measured in
effective mint intervals $1/((1-\rho)\lamintro)$, in which a window of\/ $T$ mint
intervals has length $(1-\rho)T$. At $\rho=1$ no name is ever usable. In terms of\/ $\beta=\lamburn/\lamintro$, the time-average frontier is $(1-\rho)$ times its
value at $\rho=0$, and the interval frontier is at least $(1-\rho)$ times its
value at $\rho=0$.
\end{proposition}

Appendix~\ref{app:ext} gives the proof. At $\theta=\infty$, $\bstar=0.510$, $0.352$,
$0.185$ and $0.078$ for $\rho=0.25$, $0.5$, $0.75$ and $0.9$, and no $\beta$
suffices at $\rho=1$. Faster minting lowers $\beta_{\mathrm{eff}}$ only
through names not exposed at mint.

\paragraph{Strategic block timing}
Censors can delay and batch blocks instead of blocking on discovery, as the GFW
apparently did with the Tor Browser default bridges of each release in
2015--16~\cite{censors_delay}. A censor that discovers live names at the same
aggregate rate $\lamburn$ but withholds every block until all names in use are
discovered, and then blocks them at once (strike-all), lowers $\bstar$ from
$0.661$ to $0.064$ and $\bstaravg$ from $0.840$ to $0.313$. Fixed batches of
two or four names barely move $\bstar$ ($0.671$ and $0.647$), so the damage
comes from emptying the pool at once. A defense is a reserve of $r$ of the live
names, unexposed at mint and served only once all names in use are blocked. It
has no effect on the constant-rate censor, but against strike-all it raises
$\bstar$ to $0.237$, $0.367$ and $0.522$ for $r=1$, $2$ and $4$. An average-reward
Markov decision process (MDP) gives strike-all as the exact time-average best
response at all 17 points solved, for a censor that knows how many names are in
use and a defender with no reserve. We do not compute the interval-optimal
policy, so $0.064$ is an upper bound on the interval frontier against the worst
timing. The frontiers of Sections~\ref{sec:frontier} and~\ref{sec:rotation}
therefore hold against constant-rate censors. They are not worst-case
guarantees.

\paragraph{Correlated takedowns and diversification}
Names can fail together, for example when a provider suspends an account. Let
takedowns at rate $\lamburn/b$ each burn $b$ names, so the nominal burn rate stays
$\lamburn$. Bursts lower the frontier to $\bstar=0.523$, $0.314$ and $0.066$ for
$b=2$, $4$ and $8$. Spreading names over $P$ providers confines a takedown to one
provider. With pooled minting, let takedowns at total rate $\lamburn$ each empty a
random nonempty provider, so that $\beta$ here counts takedowns, not burned
names. Splitting the $\kmax=8$ names evenly over $P=1$, $2$, $4$ and $8$ providers
gives $\bstar=0.008$, $0.092$, $0.311$ and $0.661$. Diversification recovers the
single pool with unit burns when each provider holds one name, and never exceeds
it, because on common clocks each takedown removes at least one name from the
diversified pool but exactly one from the single pool (Lemma~\ref{lem:diversify}).
It gives no help against a censor that strikes all providers at once.

\begin{takeaway}
In summary, the frontiers assume a censor that blocks each name on
discovery. Exposure at mint, timed blocks and correlated takedowns can lower them.
Reserves soften timed blocks, and spreading names over providers softens
takedowns.
\end{takeaway}

\section{Grounding the Model in Public Data}\label{sec:data}

We make no new measurements and draw on the literature and on public data,
such as the measurements of the Open Observatory of Network Interference
(OONI)~\cite{ooni}, each cited where it is used. Appendix~\ref{app:data}
documents each dataset and how we handled the dumps of Russia's blocking
registry.

\subsection{Which Name Each Censor Burns}\label{sec:burn}

Hosting platforms give each customer (a \emph{tenant}) a name under a shared
suffix, such as \texttt{<tenant>.github.io}. A censor can block one tenant name
or the whole suffix, and under a suffix-wide rule, names minted on the platform
are blocked on arrival and add nothing to $\lamintro$. Table~\ref{tab:burn}
classifies 19 platform suffixes for China, Iran and Russia. Russia's registry
entries are legal orders, and OONI measures too few platform entries to confirm
their enforcement.

\begin{table}[t]
\centering
\footnotesize
\caption{Observed or ordered granularity of name blocks under 19 platform
suffixes,
from GFWatch and OONI data for China, IRBlock scans for Iran (block onsets
unknown), and registry dumps for Russia (Appendix~\ref{app:data}). Registry
entries are legal orders, not proof of enforcement.}
\label{tab:burn}
\begin{tabular}{@{}lccc@{}}
\toprule
Platform suffix & China & Iran & Russia \\
\midrule
\texttt{github.io} & $\circ$ & $\circ$ & $\circ$ \\
\addlinespace[2pt]
\texttt{pages.dev} & $\circ$ & $\bullet$ & $\circ{\to}\bullet$\,'24$^{\dagger}$ \\
\texttt{workers.dev} & $\circ{\to}\bullet$\,'22 & ? & $\circ$ \\
\texttt{r2.dev} & ? & ? & $\varnothing$ \\
\texttt{trycloudflare.com} & ? & $\bullet$ & $\varnothing$ \\
\addlinespace[2pt]
\texttt{vercel.app} & $\circ{\to}\bullet$\,'22 & ? & $\circ$ \\
\texttt{netlify.app} & $\circ$ & $\circ$ & $\circ$ \\
\texttt{herokuapp.com} & $\circ{\to}\bullet$\,'22$^{\ddagger}$ & $\bullet$ & $\circ$ \\
\texttt{onrender.com} & $\circ$ & ? & $\circ$ \\
\texttt{fly.dev} & $\circ$ & ? & $\varnothing$ \\
\texttt{deno.dev} & $\circ$ & ? & $\circ$ \\
\addlinespace[2pt]
\texttt{appspot.com} & $\bullet$\,${\le}$'20 & ? & $\circ$ \\
\texttt{web.app} & $\circ$ & ? & $\circ$ \\
\texttt{firebaseapp.com} & $\circ$ & ? & $\circ$ \\
\texttt{run.app} & ? & ? & $\circ$ \\
\addlinespace[2pt]
\texttt{cloudfront.net} & $\circ$ & $\circ$ & $\circ$ \\
\texttt{lambda-url.*.on.aws} & ? & ? & $\varnothing$ \\
\addlinespace[2pt]
\texttt{azurewebsites.net} & $\circ$ & $\circ$ & $\circ$ \\
\texttt{azurefd.net} & $\circ$\,$\triangle$ & ? & $\circ$ \\
\bottomrule
\end{tabular}
\par\smallskip
\parbox{\columnwidth}{
$\circ$~Per-name (tenant) blocks only. For China, the cells of
\texttt{pages.dev}, \texttt{netlify.app}, \texttt{onrender.com},
\texttt{fly.dev}, \texttt{deno.dev}, \texttt{firebaseapp.com} and
\texttt{azurewebsites.net} rest on GFWatch data to September 2024 alone, as OONI tested fewer than five tenant names per suffix in 2026.
$\bullet$~Suffix-wide rule (in Russia, a registry entry).
$\circ{\to}\bullet$~Per-name blocks, then a suffix-wide rule. Years give the
rule's onset (for China, in GFWatch), and ${\le}$'20 means in place when those
data begin in March 2020.
$\triangle$~Temporary impairment of tested names, trigger undetermined.
?~Too few tenant names known to be tested, and no block seen.
$\varnothing$~No registry entry. Blocking outside the registry is not covered.
$^{\dagger}$Wildcard entry, enforcement unconfirmed by OONI.
$^{\ddagger}$In 2026 the bare suffix and Heroku's routing names are largely not
blocked.}
\end{table}

Some platforms are blocked one tenant at a time. In China, GFWatch recorded 259
blocked \texttt{github.io} tenants but no suffix rule, and in 2026 OONI found 6
of 980 tested names blocked~\cite{gfwatch,ooni_data}. Iran blocked 27
\texttt{github.io} tenants by name~\cite{irblock}, and Russia's registry listed
30 \texttt{cloudfront.net} tenants at the end of 2021 and 10,224 in October
2025~\cite{zapret_info}. Whole platforms have
also been blocked. China's DNS rules covered \texttt{workers.dev},
\texttt{vercel.app} and \texttt{herokuapp.com} in 2022, at least 405, 42 and
771 days after their first tenant blocks~\cite{gfwatch}, and in 2026 OONI still
found 984 of 1,100, 1,008 of 1,008 and 357 of 376 tested names under them
blocked~\cite{ooni_data}. The blockable unit thus depends on the censor and the
platform, and an escalation from tenant names to a suffix rule burns every live
name on the platform at once, a burst in the sense of Section~\ref{sec:ext}. A
deployment must therefore count $\lamintro$ and $\lamburn$ in the units its
censor blocks, platform by platform, and expect that granularity to change.

\subsection{Names in a Peer-to-Peer System}\label{sec:snowflake}

Snowflake~\cite{snowflake} shows that even a peer-to-peer design whose
proxy addresses churn depends on names at rendezvous. Tor's broker counted a median
of 148,719 unique proxy addresses per day from 28 September to 3 October
2026~\cite{tor_collector}, and half the 2023 proxy pool turned over in about 20
hours~\cite{snowflake}. To reach a proxy, however, a client must first poll
Tor's broker, and every method the broker counts uses a censor-visible
name, namely a CDN front or the broker's own domain, Google's Accelerated
Mobile Pages (AMP) cache, or Amazon's Simple Queue Service
(SQS)~\cite{snowflake,tor_collector}.

When Turkmenistan blocked the default broker front for months from October
2021, estimated Snowflake users there fell from 28 to
2~\cite{snowflake,tor_metrics}.
When Fastly ended domain fronting on 1 March 2024~\cite{tor_metrics_timeline},
the Fastly fronts in use failed together. Within a week, HTTP rendezvous polls
fell 85.8\% and estimated users 32.5\%, while proxy addresses fell only
2.1\%~\cite{tor_collector,tor_metrics} (Appendix~\ref{app:data} defines the
baselines). Tor configured other fronts the same day~\cite{tor_rdsys_admin},
yet HTTP polls regained half their earlier volume only after 36
days~\cite{tor_collector}. Tor thus runs a small name economy for rendezvous.
Its Connection Assist settings introduced 19 rendezvous names and retired 11
from 2022 to 2026~\cite{tor_rdsys_admin}. In the model's terms, $\kmax$ is
small, minting means changing a configuration or, for built-in defaults,
shipping a release, and the Fastly event is a correlated takedown
(Section~\ref{sec:ext}). A client needs a live name only for new proxies, so
session survival lies
between $A_T$ and $A$.

\subsection{Minting Economics and Exposure}\label{sec:params}

On 3 October 2026, one registrar listed 36 top-level domains (TLDs) under \$2
for the first year, the cheapest at \$1.54~\cite{porkbun_pricing}. Let's Encrypt
issues free certificates and allows 2,400 orders a day per
account~\cite{letsencrypt_rate_limits}. Labels under an owned
domain, and tenant names within a platform's per-account cap, count as free new
names if the censor blocks exact names (Section~\ref{sec:burn}).

Minting can, however, expose names. Publicly trusted certificates must be logged
in CT for Chrome and Apple platforms to accept
them~\cite{chrome_ct_policy,apple_ct_policy}, the logs commit to adding entries
within 60\,s or 24~hours~\cite{chrome_ct_log_list,rfc6962}, and in 2018 third
parties looked up new names 73 to 197\,s after logging~\cite{ct_rise}. New
names in generic TLDs also appear in daily zone files~\cite{icann_czds}. Of the
unit types we examined, only three are unexposed at mint, namely labels under a
wildcard certificate, unpublished tenant names under a platform wildcard (as for
\texttt{*.netlify.app} and \texttt{*.web.app}), and labels with a private pinned
certificate, which WebTunnel has supported since 2025~\cite{tor_blog_2025}.
The exposed share $\rho$ is thus near 1 for registrable generic-TLD domains and
CT-logged names and near 0 for names that no channel exposes at mint. The rate
$\theta$ at which a censor burns exposed names is unmeasured. If the censor
acts promptly, only unexposed mints raise the effective mint rate
(Proposition~\ref{prop:exposure}).

\section{Implications for Operators}\label{sec:implications}

For a name-dependent deployment (Section~\ref{sec:threat}), the results give
one principle. Make minting, not rotation, the first-class control.
Rotation must be fast enough to keep the address layer up through whole
sessions (Section~\ref{sec:rotation}). Beyond that speed the burn ratio $\beta$
decides whether the target is met, so operators should track $\beta$ as a health
metric, with an alarm as it nears $\bstar$, or
$(1-\rho)\bstar$ when a share $\rho$ of mints is exposed (step~4).

For a nameless deployment, steps~1 and~2 are the whole procedure
(Section~\ref{sec:rotation}).

\begin{enumerate}[leftmargin=*]
\item Fix $(\alpha,T)$ in wall-clock terms. With rates per unit of wall-clock
time, a session of length $T_{\mathrm{wall}}$ lasts $T=\lamintro T_{\mathrm{wall}}$
mint intervals and $\lama T_{\mathrm{wall}}$ mean discovery delays. Sessions
must be able to move between endpoints and names (Definition~\ref{def:avail}).

\item Size the address layer for the whole window, not an instant. Measure
the blocked fraction $p$ at the rotation rate in use. At a new rate $\mu$,
recompute $p=\mathbb{E}[e^{-\mu D}]$ from the measured law of the discovery
delay~$D$, not from its mean alone (Remark~\ref{rem:delay}). Choose $n$ and
$\mu$ so that $1-p^n(1+n\mu T_{\mathrm{wall}})$ is close to~1. This lower bound
on the address
window $W^{\mathrm{addr}}_T$ holds for any delay law under Poisson rotation
(Corollary~\ref{cor:window}). Making $p^n\le1-\alpha$ (Lemma~\ref{lem:addr}) is
not enough, because outages inside the window grow with $\lama
T_{\mathrm{wall}}$ and the name
layer needs the remaining margin (Section~\ref{sec:rotation}). Discovery must
also be independent across endpoints (Remark~\ref{rem:sync}), and $n$ counts
only the endpoints to which a client's sessions can move.

\item Choose the buffer $\kmax$ and compute $\bstar(\alpha,T)$ for the actual
$(n,\mu/\lamintro,\lama/\lamintro,\kmax,T)$ with the exact solver
(Proposition~\ref{prop:exact},
Appendix~\ref{app:exact}) in about 8\,ms per frontier at $\kmax=8$.
The solver's address window assumes exponential discovery. Otherwise use the
bound of step~2. If
takedowns remove $b$ names at once, use the frontier for that~$b$
(Section~\ref{sec:ext}).

\item Mint at
\[
  \lamintro \ge \frac{\lamburn}{(1-\rho)\,\bstar(\alpha,\lamintro T_{\mathrm{wall}})},
\]
where $\rho$ is the fraction of mints exposed at mint time
(Section~\ref{sec:params}). With exposed names burned at once, the rule is
conservative, and at $\rho=1$ no mint rate suffices
(Proposition~\ref{prop:exposure}). Bisection in $\lamintro$ solves the rule,
since more mints never leave fewer
live names (Lemma~\ref{lem:coupling}). Count names at the censor's granularity
(Table~\ref{tab:burn}), so that all hostnames under a blocked suffix or
registrable domain form one name.

\item Keep a reserve of $r$ of the live names, unserved and unexposed until the
active pool is empty. It recovers part of the frontier lost to a strike-all
censor and leaves it unchanged against a constant-rate censor
(Section~\ref{sec:ext}). If every strike empties the active pool, use the
strike-all frontier for the chosen~$r$ in steps~4 and~7, not the burst frontier
of step~3.

\item Prefer platforms whose tenant names the censor blocks one at a time
(Table~\ref{tab:burn}), within their quotas and terms of service. Spreading
names over independent providers moves the frontier toward that of unit burns
but never beyond it, with $\beta$ counting takedowns (Section~\ref{sec:ext}). It
does not slow a censor whose discovery spans providers, for example through CT.

\item Estimate the rates online in each censored region. Burns occur at rate
$\lamburn$ only while $K\ge1$, so the maximum-likelihood estimate is
\[
  \hat{\lambda}_{\mathrm{burn}}=\frac{\text{burns of names in use}}{\text{time with }K\ge1}.
\]
It is consistent if each burn is observed when it occurs, and missed or late
detections bias it low. Burns divided by elapsed time instead estimate
$\lamburn(1-\pi_0)$, which hides $\beta>1$ because $\beta(1-\pi_0)$ never
exceeds~1 (Theorem~\ref{thm:cap}). Also record~$b$. Names burned
before first use reveal exposure at mint, so their share of mints bounds
$\rho$ from below. Censors differ in what they burn (Table~\ref{tab:burn}), so
meet step~4 in every region, and redo steps~3 and~4 when
$\hat{\lambda}_{\mathrm{burn}}/\lamintro$ nears~$(1-\rho)\bstar$.
\end{enumerate}

\section{Discussion}\label{sec:discussion}

\paragraph{What the censor can do}
Most countermoves change a model parameter and have a defense
(Sections~\ref{sec:ext} and~\ref{sec:implications}). Escalation from
tenant blocks to a suffix rule (Section~\ref{sec:burn}), limited by spreading
names over platforms, and range blocking in a crisis are regime changes. The
model thus turns ``the censor will adapt'' into parameters to measure and
defenses to size, although the censor's best timing against interval
availability remains open.

\paragraph{Generality and future work}
Nothing in the model is specific to proxies, so it applies to any service on
rotating infrastructure whose clients present a censor-visible name, such as a
news mirror or a messaging entry point. Its rates are dimensionless
(Section~\ref{sec:rotation}), so only the wall-clock calibration changes from
one service to another. The rotation layer analyzed here is the foundation
of a complete system, which also needs a discovery layer that resists
enumeration and a traffic-shaping layer that resists fingerprinting. Both build
on the rotation layer, and extending the timing model
to them is future work.

\paragraph{Limits}
We do not estimate $\beta$ for any deployed system, because public sources
record which names were blocked, not when a deployment minted them, and test
popular or listed names. Measuring $\beta$ needs mint dates and repeated block
observations of the same names, from data such as access-gated zone files of
ICANN's Centralized Zone Data Service~\cite{icann_czds} or unpublished GFWatch
tests~\cite{gfwatch}. We use none of them, so $\beta$ and $\rho$ remain swept
parameters for operators to measure (Section~\ref{sec:implications}). Our data
are snapshots of particular censors, periods and vantage points
(Appendix~\ref{app:data}), and censors may since have changed what they block.
Protocol fingerprinting (Section~\ref{sec:threat})
can only lower availability at fixed address and name rates, so against a
censor that also fingerprints, the model's values are upper bounds.

\section{Conclusion}\label{sec:conclusion}

We model moving-target circumvention for deployments whose clients must
present a name the censor can block. Redundant endpoints rotated fast enough
relative to discovery keep the address layer up, provided the censor finds them
one at a time, and past that speed faster rotation barely moves the frontier.
The name layer then binds, because the burn ratio $\beta$ keeps time-average
availability below $1/\beta$ however fast endpoints rotate. The sustainability
frontier $\bstar(\alpha,T)$ tells an operator how fast to mint, and bursty
takedowns, timed blocks and burns of names exposed at mint can lower it. Public
data show that censors burn different units on different platforms, and that
even a peer-to-peer system depends on names at rendezvous. How fast censors
burn privately minted names, and hence $\beta$ in the field, remains to be
measured. Still, the advice is simple. Once rotation is fast enough, mint
names to keep $\beta$ below $\bstar$, and track $\beta$ as a health
metric.\label{sec:bodyend}

\appendix

\begin{ai}
The model and the thesis are the authors' work. AI-based tools (Claude Code)
were used, under the authors' direction, to implement the simulation, analysis
and plotting code. We
have manually verified and are responsible for the accuracy, originality, and
integrity of the output of all AI-based tools.
\end{ai}

\clearpage

\bibliographystyle{ACM-Reference-Format}
\bibliography{reference}

\section{Proofs}\label{app:proofs}

This appendix gives two auxiliary lemmas and proves Lemma~\ref{lem:addr},
Theorems~\ref{thm:closed} and~\ref{thm:frontier}, and
Propositions~\ref{prop:pool} and~\ref{prop:conserve}. Appendix~\ref{app:exact} proves
Proposition~\ref{prop:exact}, and Appendix~\ref{app:ext} derives the
extensions. All rates are positive, $\kmax\ge1$, and time is measured in mean
mint intervals $1/\lamintro$. Bold symbols denote vectors. Every chain below is
finite and irreducible, so
its unique stationary law gives the long-run fraction of time in each set of
states.

\begin{proof}[Proof of Lemma~\ref{lem:addr}]
An unblocked endpoint becomes blocked at rate $\lama$, and a blocked one
becomes unblocked at rate $\mu$, since its next rotation moves it to an
unblocked address. This two-state Markov chain is blocked with stationary
probability $\lama/(\lama+\mu)$. The $n$ chains are driven by independent
Poisson processes, so they are independent and all $n$ are blocked with
probability $p^{n}$.
\end{proof}

\begin{proof}[Proof of Proposition~\ref{prop:pool}]
Mints and burns occur at the constant rates of Definition~\ref{def:game}, so
$K$ is a birth--death chain~\cite[Section~3.6]{kleinrock}. Detailed balance,
$\pi_k\lamintro=\pi_{k+1}\lamburn$, gives $\pi_{k+1}=\pi_k/\beta$, and
normalizing the geometric series gives~\eqref{eq:pi0}. The drift is
$\lamintro-\lamburn=\lamintro(1-\beta)$. For $\beta<1$, \eqref{eq:pi0} equals
$(1-\beta)\beta^{\kmax}/(1-\beta^{\kmax+1})$, which is at most $\beta^{\kmax}$.
\end{proof}

\begin{lemma}[Empty-pool probability]\label{lem:pi0}
Let $G_k(\beta)=\sum_{i=0}^{k}\beta^{-i}$, so that $\pi_0(\beta,k)=1/G_k(\beta)$
for a pool with buffer $k$ (Proposition~\ref{prop:pool}). The following hold
for $\beta>0$ and $k\ge1$.
\begin{enumerate}[label=\textup{(\alph*)},leftmargin=*]
\item $\pi_0$ is continuous and strictly increasing in $\beta$, and it tends to
  $0$ as $\beta\to0$ and to $1$ as $\beta\to\infty$.
\item $\pi_0$ is strictly decreasing in $k$.
\item $1-\pi_0(\beta,k)=\beta^{-1}G_{k-1}(\beta)/G_k(\beta)<\min\{1,1/\beta\}$.
\item As $k\to\infty$, $1-\pi_0(\beta,k)$ increases to $\min\{1,1/\beta\}$.
\end{enumerate}
\end{lemma}

\begin{proof}
For $i\ge1$ the term $\beta^{-i}$ is continuous and strictly decreasing in
$\beta$, from $\infty$ to $0$, which gives (a). Part (b) follows from
$G_{k+1}=G_k+\beta^{-(k+1)}$. For (c), $1-1/G_k=(G_k-1)/G_k$ and
$G_k-1=\beta^{-1}G_{k-1}<\beta^{-1}G_k$. For (d), $G_k$ increases to
$\beta/(\beta-1)$ if $\beta>1$ and to $\infty$ if $\beta\le1$.
\end{proof}

\begin{lemma}[Monotone coupling]\label{lem:coupling}
Let $K$ and $K'$ be birth--death chains on $\{0,\dots,k\}$ and
$\{0,\dots,k'\}$, with $k'\le k$, birth rates $\eta_j$ and $\eta'_j$ (from state
$j$), and death rates $\xi_j$ and $\xi'_j$ (from state $j\ge1$). Suppose that
$\eta'_j\le\eta_j$ for $j<k'$ and $\xi'_j\ge\xi_j$ for $1\le j\le k'$, and start
each chain in its stationary law. Then the chains can be coupled so that
$K'(t)\le K(t)$ for all $t\ge0$, and for every $T\ge0$,
\[
\Pr\bigl[K'(t)\ge1\text{ for all }t\in[0,T]\bigr]\le\Pr\bigl[K(t)\ge1\text{ for all }t\in[0,T]\bigr].
\]
At $T=0$ this compares the stationary probabilities of a nonempty pool.
\end{lemma}

\begin{proof}
For the stationary laws $\boldsymbol{\pi}$ of $K$ and $\boldsymbol{\pi}'$ of $K'$, detailed balance gives $\pi_{j+1}/\pi_j=\eta_j/\xi_{j+1}$ and
$\pi'_{j+1}/\pi'_j=\eta'_j/\xi'_{j+1}$. Hence $\pi_j/\pi'_j$ is nondecreasing on
$\{0,\dots,k'\}$, and $\pi'_j=0$ for $j>k'$. This likelihood-ratio order implies
the usual stochastic order, so the initial states can be drawn with
$K'(0)\le K(0)$. While $K'\ne K$, let the chains move on independent clocks.
While $K'=K=j$, let one clock of rate $\eta'_j$ move both up and another of rate
$\eta_j-\eta'_j$ move only $K$ up (with $\eta'_{k'}=0$). Likewise, let one clock
of rate $\xi_j$ move both down and another of rate $\xi'_j-\xi_j$ move only $K'$
down. Each chain then jumps at its own rates, so each is a stationary copy of
its law. A single jump of size one cannot reverse a strict order, and at
equality only $K$ can rise alone and only $K'$ can fall alone. Hence
$K'(t)\le K(t)$ for all $t$, and the event on the left is contained in the
event on the right.
\end{proof}

\begin{proof}[Proof of Theorem~\ref{thm:closed}]
The address layer (state $c$) and the name layer (state $K$) are driven by
disjoint sets of independent Poisson clocks, and no transition of one depends
on the state of the other. The joint chain $(c,K)$ therefore has the generator
$Q=Q_c\otimes I+I\otimes Q_K$, the Kronecker sum of the generators of the
layers. The product of their stationary laws $\boldsymbol{\pi}_c$ and
$\boldsymbol{\pi}_K$ satisfies
\[(\boldsymbol{\pi}_c\otimes\boldsymbol{\pi}_K)Q=(\boldsymbol{\pi}_cQ_c)\otimes\boldsymbol{\pi}_K+\boldsymbol{\pi}_c\otimes(\boldsymbol{\pi}_KQ_K)=\mathbf{0},\]
and the
joint chain is irreducible, so the product is its stationary law. The deployment is reachable iff $c\ge1$ and $K\ge1$. Hence
$A=\Pr[c\ge1]\Pr[K\ge1]=(1-p^n)(1-\pi_0)$ by Lemma~\ref{lem:addr} and
Proposition~\ref{prop:pool}.
\end{proof}

\begin{proof}[Proof of Theorem~\ref{thm:frontier}]
Write $h(\beta)=1-\pi_0(\beta,\kmax)$, so that $A=c_ah(\beta)$. By
Lemma~\ref{lem:pi0}(a), $h$ decreases continuously and strictly from $1$ to $0$
as $\beta$ runs from $0$ to $\infty$.

(i)~Existence. If $c_a>\alpha$, the equation $h(\beta)=\alpha/c_a$ has a unique
root $\bstaravg$, and $A\ge\alpha$ iff $\beta\le\bstaravg$. If $c_a\le\alpha$,
then $A<c_a\le\alpha$ for every $\beta>0$, so no $\beta$ suffices.

(ii)~Monotonicity. By Lemma~\ref{lem:pi0}(b), a larger $\kmax$ raises $h$ at
every $\beta$. A larger $\mu$ lowers $p$, and a larger $n$ lowers $p^n$, so
either raises $c_a$ and lowers $\alpha/c_a$. Each change moves the root
strictly to the right.

(iii)~Bound and limit. If $c_a\le\alpha$, then $\bstaravg=0$. Otherwise, at
the root, $\alpha/c_a=h(\bstaravg)<1/\bstaravg$ (Lemma~\ref{lem:pi0}(c)), so
$\bstaravg<c_a/\alpha\le1/\alpha$. As $\kmax$ grows, the roots increase and
stay below $c_a/\alpha$. Suppose their limit $L$
were below $c_a/\alpha$, and take $\beta'$ in $(L,c_a/\alpha)$. Then
$h(\beta')\to\min\{1,1/\beta'\}>\alpha/c_a$ (Lemma~\ref{lem:pi0}(d)), so for
large $\kmax$ the root would exceed $\beta'>L$, a contradiction. Hence
$\bstaravg\to c_a/\alpha$.

(iv)~Sustainability frontier. By Proposition~\ref{prop:exact},
$A_T=W^{\mathrm{addr}}_T\,W^{\mathrm{name}}_T$, the product of the address
window and the name window. A larger $\lamburn$ raises every death rate of the
pool, so Lemma~\ref{lem:coupling} shows that $W^{\mathrm{name}}_T$, and with it $A_T$,
is nonincreasing in $\beta$. Each is also continuous in $\beta$. The set of
$\beta$ with $A_T\ge\alpha$ is therefore an interval $(0,\bstar(\alpha,T)]$,
possibly empty. Since $A_T\le A$,
$\bstar(\alpha,T)\le\bstaravg$. Raising $\kmax$ by one adds a state without
changing the other rates, so Lemma~\ref{lem:coupling}, with $k'=\kmax$ and
$k=\kmax+1$, shows that $\bstar(\alpha,T)$ is nondecreasing in $\kmax$.
Likewise, a larger
$n$ or $\mu$ raises the birth rates $\mu(n-c)$ of the address chain, so
$W^{\mathrm{addr}}_T$, and with it $\bstar(\alpha,T)$, is nondecreasing in $n$
and $\mu$. In all,
\[
\bstar(\alpha,T)\le\bstaravg<c_a/\alpha\le1/\alpha.\qedhere
\]
\end{proof}

\begin{proposition}[Conservation of names]\label{prop:conserve}
Let mints arrive by any point process whose mean count grows at rate
$\lamintro$, each admitted while $K<\kmax$, and let burns, each removing at
least one live name, arrive with stochastic intensity at least $\lamburn$
whenever $K\ge1$. Then the expected fraction of\/ $[0,t]$ with $K\ge1$ is at
most $1/\beta+o(1)$ as $t\to\infty$, and in stationarity
$A\le\Pr[K\ge1]\le1/\beta$.
\end{proposition}

\begin{proof}
Each burn removes a name that was admitted earlier or present at time $0$, so
the number $B(t)$ of burns in $[0,t]$ is at most the number $M(t)$ of admitted
mints plus $\kmax$. By the intensity assumption,
$\lamburn\,\mathbb{E}\bigl[\int_0^t\mathbb{I}\{K(s)\ge1\}\,ds\bigr]\le
\mathbb{E}[B(t)]\le\mathbb{E}[M(t)]+\kmax\le\lamintro t+o(t)$. Dividing by
$\lamburn t$ gives the first claim. In stationarity the expected fraction is
$\Pr[K\ge1]$, and a reachable deployment needs a live name.
\end{proof}

\section{Exact Interval Availability}\label{app:exact}

\begin{proof}[Proof of Proposition~\ref{prop:exact}]
Let $X$ be a finite irreducible chain with generator $Q$, stationary law $\boldsymbol{\pi}$
and up set $\mathcal{U}$, and let $\sigma$ be the first time $X$ is outside
$\mathcal{U}$. The chain killed on leaving $\mathcal{U}$ has the subgenerator
$Q_{\mathcal{U}\mathcal{U}}$, the restriction of $Q$ to $\mathcal{U}$, and the
transition matrix $e^{Q_{\mathcal{U}\mathcal{U}}t}$. Hence
$\Pr_i[\sigma>T]=(e^{Q_{\mathcal{U}\mathcal{U}}T}\,\mathbf{1})_i$ for
$i\in\mathcal{U}$, and averaging over the stationary initial state gives
\[
W_T=\Pr\nolimits_{\boldsymbol{\pi}}[\sigma>T]=\boldsymbol{\pi}_{\mathcal{U}}\,e^{Q_{\mathcal{U}\mathcal{U}}T}\,\mathbf{1},
\]
the window formula of \eqref{eq:exact}. In the model the joint chain is the
Kronecker sum of the two layers (proof of Theorem~\ref{thm:closed}), and
$\mathcal{U}=\{c\ge1\}\times\{K\ge1\}$. The restricted generator is the
Kronecker sum of the restrictions
$Q_1$ and $Q_2$ of $Q_c$ and $Q_K$ to their up sets, with
$e^{(Q_1\otimes I+I\otimes Q_2)T}=e^{Q_1T}\otimes e^{Q_2T}$, and
$\boldsymbol{\pi}_{\mathcal{U}}$ is a product. So $A_T=W^{\mathrm{addr}}_T\,W^{\mathrm{name}}_T$,
the window formula applied to each layer.
\end{proof}

\begin{remark}[General discovery delays]\label{rem:delay}
Keep rotation Poisson at rate $\mu$, but let each new address be discovered
after an independent delay $D$ with any law. In an exponential rotation interval $V$
of mean $1/\mu$, the address is blocked for $(V-D)^{+}$, so an
endpoint is blocked with probability
$p=\mu\,\mathbb{E}[(V-D)^{+}]=\mathbb{E}[e^{-\mu D}]=\Pr[D<V]$, and
$c_a=1-p^{n}$ still holds for independent endpoints. At a given mean, a
deterministic delay minimizes $p$, because
$\mathbb{E}[e^{-\mu D}]\ge e^{-\mu\,\mathbb{E}[D]}$. A mixture of exponentials
gives at least the value $1/(1+\mu\,\mathbb{E}[D])$ of a single exponential.
In general $p$ is
fixed by neither the mean nor the tail of $D$. With $\mathbb{E}[D]=1/\lama$ and
$\mu=3\lama$, the exponential gives $p=0.25$ and a deterministic delay $0.05$,
but a light-tailed hyperexponential law that is usually fast (mean
$0.1/\lama$, weight $0.9$) and occasionally slow (mean $9.1/\lama$, weight
$0.1$) gives $0.70$, which leaves $n=8$ endpoints with $c_a=0.945$. The
exponential model is therefore not conservative, and an operator should size
$n$ from the measured blocked fraction $p$, not from a mean discovery delay.
\end{remark}

\begin{corollary}[A closed-form lower bound]\label{cor:window}
Let a stationary process be up with probability $\gamma$, and let $\omega$ be the
stationary rate of its transitions from up states to down states. It stays up
throughout a window of length $T$ with probability at least $\gamma-\omega T$. For
the two layers of the model,
\[
W^{\mathrm{addr}}_T\ge1-p^n(1+n\mu T),\qquad
W^{\mathrm{name}}_T\ge1-(1+\lamintro T)\,\pi_0.
\]
The address bound holds under Poisson rotation for every discovery-delay law of
Remark~\ref{rem:delay}.
\end{corollary}

\begin{proof}
A window fails if the process is down at its start (probability $1-\gamma$) or
enters the down set during the window, and the expected number $\omega T$ of
such entries bounds the probability of at least one. In the setting of
Remark~\ref{rem:delay}, an endpoint becomes blocked once in every
rotation interval with $D<V$, that is, at rate $\mu p$. Since the
endpoints are independent, the all-blocked state is entered at rate
$n\cdot\mu p\cdot p^{n-1}=n\mu p^n$. For the pool, entries into $K=0$ come from
$K=1$ at rate $\lamburn$, and $\lamburn\pi_1=\lamintro\pi_0$.
\end{proof}

The address bound needs only the measured $p$ and the mean number $\mu T$ of
rotations per endpoint in a window, so it gives a closed-form check when sizing $n$
(Section~\ref{sec:implications}). At the canonical configuration it is 0.99815,
against the exact $W^{\mathrm{addr}}_5=0.99825$. Both bounds lose accuracy as
$\mu T$ or $\lamintro T$ grows.

\paragraph{Numerical method}
Both layers are birth--death chains, which are reversible. With
$\Pi=\mathrm{diag}(\boldsymbol{\pi}_{\mathcal{U}})$,
the matrix $S=\Pi^{1/2}Q_{\mathcal{U}\mathcal{U}}\Pi^{-1/2}$ is then symmetric
and tridiagonal, and \[
W_T=\mathbf{y}^{\top}\,e^{ST}\,\mathbf{y}
=\sum_m(\mathbf{z}_m^{\top}\mathbf{y})^2e^{\nu_mT},
\]
where
$\mathbf{y}=\Pi^{1/2}\mathbf{1}$ and $(\nu_m,\mathbf{z}_m)$ are the orthonormal
eigenpairs of
$S$. The weights are nonnegative, so the sum involves no cancellation. Above 200
states, SciPy's \texttt{expm\_multiply}, a scaled truncated Taylor method,
replaces the eigendecomposition. The non-birth--death chains of the extensions
and the sweep
chain of Remark~\ref{rem:sync} use an LU solve for stationary laws and a dense
matrix exponential up to 400 states, with \texttt{expm\_multiply} above.
Frontiers are Brent roots (tolerance $10^{-10}$), and each root is the
frontier, by the monotonicity of availability in $\beta$ for the birth--death
pools (Lemma~\ref{lem:coupling}) and by a single-crossing grid check for the
other chains.

\paragraph{Accuracy}
Window probabilities are accurate to about $10^{-12}$. On 40 random
birth--death chains the symmetric form agrees with a dense matrix exponential
to $2\times10^{-14}$, and on constant-rate name pools with buffers of 300 to
1,200 it agrees with \texttt{expm\_multiply} to $4.4\times10^{-15}$. An
independent implementation, which uses different algorithms for stationary
laws (elimination), window probabilities (uniformization) and roots
(bisection), reproduces the 120 interval frontiers for
$\alpha\in\{0.9,0.95,0.99\}$, $T\in\{1,2,5,10,20\}$ and
$\kmax\in\{2,4,8,\dots,256\}$ to $2.5\times10^{-11}$.

\paragraph{Implementation check}
The simulator (Appendix~\ref{app:sim}) and the exact solver implement the
model independently. Across nine simulation experiments and a grid of 65 frontiers computed from
the closed form (about 2,800 values in all), they agree within Monte Carlo error where standard errors exist and to
within 0.009 elsewhere. In a 148-point sweep of interval availability at $T=5$
against $\beta$ ($\kmax\in\{4,8,16,32\}$, 24 seeds), the largest difference is
0.0052, and 99.3\% of the differences lie within three standard errors. Both
assume independent layers, so their agreement checks the implementations, not
the product form or the model's realism.

\section{Extension Models}\label{app:ext}

Each extension changes one ingredient of the model,
keeps the canonical address layer and is a finite chain solved exactly
(Appendix~\ref{app:exact}). Independent Monte Carlo simulations also checked
the chains for exposure, strategic timing, providers
and the sweeps of Remark~\ref{rem:sync} (16 seeds per configuration,
all differences within 2.4 standard errors).

\paragraph{Mint-time exposure}
The chain counts exposed and unexposed live names, $K_e$ and $K_u$. A mint, at
rate $\lamintro$ while $K_e+K_u<\kmax$, is exposed with probability $\rho$. The
censor burns each exposed name preemptively at rate $\theta$, and usage burns
arrive at aggregate rate $\lamburn$ while $K_e+K_u\ge1$ and hit a uniformly
chosen live name. The chain is up iff $K_e+K_u\ge1$. At nine tested points
($\kmax=8$, $T=5$, $\rho\in\{0.25,0.5,0.9\}$, $\beta\in\{0.1,0.3,0.6\}$), the
window probability differs from its limit as $\theta\to\infty$ by at most
$1.5\times10^{-4}$ at $\theta=10^3$ and $1.5\times10^{-5}$ at $\theta=10^4$.

\begin{proof}[Proof of Proposition~\ref{prop:exposure}]
In the limit, unexposed mints form a thinned Poisson process of rate
$(1-\rho)\lamintro$, and usage burns hit only unexposed names. The count $K_u$
is then the pool of Proposition~\ref{prop:pool} with mint rate
$(1-\rho)\lamintro$, whose stationary law depends only on
$\beta_{\mathrm{eff}}$. This gives the time-average statement. In units of
$1/((1-\rho)\lamintro)$, a window of $T$ mint intervals has length $(1-\rho)T$,
so the name window is that of an unexposed pool at burn ratio
$\beta_{\mathrm{eff}}$ over $(1-\rho)T$, which is at least its value over $T$.
Every $\beta\le(1-\rho)\bstar$, with $\bstar$ the frontier at $\rho=0$, has
$\beta_{\mathrm{eff}}\le\bstar$ and so meets the target by
Lemma~\ref{lem:coupling}, since the address layer is unchanged. At $\rho=1$ no
mint enters the usable pool.
\end{proof}

At $\kmax=8$ the bound $(1-\rho)\bstar$ gives 0.496, 0.331, 0.165 and 0.066 for
$\rho=0.25$, 0.5, 0.75 and 0.9, 3\% to 15\% below the exact frontiers of
Section~\ref{sec:ext}, so the mint-rate rule of Section~\ref{sec:implications}
is conservative at $\theta=\infty$.

\paragraph{Strategic block timing}
The state $(K,R,J)$ counts $K$ active names in use, $J\le K$ of them discovered
but not yet blocked, and $R\le r$ reserve names, never served and so
undiscovered. A mint, at rate $\lamintro$ while $K+R<\kmax$, joins the reserve
if $K\ge1$ and $R<r$, and the active pool otherwise. When the active pool
empties, the whole reserve becomes active at once. The censor discovers
undiscovered active names at aggregate rate $\lamburn$ while $J<K$, and the chain is up iff $K\ge1$. The constant-rate policy blocks each name at discovery,
which makes $K+R$ the constant-rate pool of Proposition~\ref{prop:pool}, on
which a reserve has no effect. The batch policy of size $j$ blocks all
discovered names when $J=j$, and strike-all when $J=K$, emptying the active
pool. With $r=0$, minimizing time-average availability is an average-reward
MDP for a censor that observes $(K,J)$ and may block any number of discovered
names after each event. Relative value iteration finds
strike-all optimal at all 17 points tested ($\beta$ from 0.05 to 2.0 at
$\kmax=8$ and from 0.1 to 1.5 at $\kmax=32$, plus each buffer's time-average
strike-all frontier). Intuitively, an early block only frees a slot that the
defender refills with an undiscovered name, whereas blocking a fully
discovered active pool empties it.

\paragraph{Bursts and providers}
A takedown, at rate $\lamburn/b$, moves the pool from $K$ to $\max\{0,K-b\}$
names, so the nominal burn rate stays $\lamburn$. With $P$ providers of
$k_p=\kmax/P$ names each and pooled minting, a mint, at total rate $\lamintro$
while fewer than $\kmax$ names are live, goes to a uniformly chosen provider
that is not full. Takedowns, at total rate $\lamburn$, each empty a uniformly
chosen nonempty provider, so $\beta=\lamburn/\lamintro$ counts takedowns, each
removing up to $k_p$ names. With $P=1$ the model is the burst model with
$b=\kmax$ at nominal burn ratio $\kmax\beta$.

\begin{lemma}[Diversification and the single pool]\label{lem:diversify}
For every $T\ge0$, the $P$-provider pool stays nonempty throughout a window of
length $T$ with probability at most that of the single pool with buffer $\kmax$
and unit burns at rate $\lamburn$. Equality holds when $k_p=1$.
\end{lemma}

\begin{proof}
Run both models from empty on a common mint clock of rate $\lamintro$ and a
common takedown clock of rate $\lamburn$, and let $K_P$ and $K$ count their
live names. A mint adds a name to each model with fewer than $\kmax$ names. A
takedown removes at least one name from a nonempty $P$-provider pool and
exactly one from a nonempty single pool. Both events preserve $K_P\le K$, so
$K_P(t)\le K(t)$ for all $t$, and the joint chain stays in the closed set of
ordered states. That set carries a stationary law whose marginals are those of
the two models, and under it $K_P\ge1$ throughout the window implies $K\ge1$
throughout it. With $k_p=1$ every takedown removes exactly one name, and the
two totals coincide.
\end{proof}

\paragraph{Per-name discovery}
A censor that finds each live name at its own rate $\delta$ burns at rate
$\delta K$, so the pool is the birth--death chain with death rate $\delta k$ in
state $k$, and its burns slow as the pool empties. At $\kmax=8$ the target
$(0.95,5)$ holds up to $\delta=0.212$ per mint interval. There the censor burns
$0.952$ names per mint interval while a name is live, the rate that step~7 of
Section~\ref{sec:implications} estimates, against $\bstar=0.661$ for a
constant-rate censor. By flow balance that rate is $\Pr[K<\kmax]/\Pr[K\ge1]$,
which rises with $\delta$, because a larger $\delta$ never leaves more live
names (Lemma~\ref{lem:coupling}). At this buffer and target, keeping that estimate below $\bstar$ is thus
conservative
against this censor too.

\section{Data Sources and Processing}\label{app:data}

Sections~\ref{sec:rotation} and~\ref{sec:data} rest on four analyses of public
data, covering burn units, Snowflake, literature parameters and minting
economics. None of them made a new measurement or used an account, key or
approval, and
all requests were paced. No human subjects are involved. The Tor and Snowflake
statistics are published aggregates. We refer to issue threads by URL and date,
to comments by position and time, and to commits by hash, never by author, and
we do not redistribute them. Each analysis has a fetch script and an analysis
script. The fetch script
downloaded the raw files on \mbox{2026-10-03} and recorded where and when each file
came from, with its hash and license. The analysis script rebuilds the derived
data offline from those files, and a rerun on the same files reproduced the
derived tables exactly. A separate reimplementation, not part of the artifact,
recomputed every number that the analyses report.

\paragraph{Registry data}
Russia's blocking registry, read through the zapret-info
mirror~\cite{zapret_info}, lists the names, URLs and addresses under blocking
orders, some of them for illegal content. Each registry file was
streamed through a filter that kept only the hostnames that lie under one of 19
platform
suffixes, with their decision dates. No URL path was kept. The full files were
hashed for provenance and
then deleted. No listed name, URL or address was resolved or fetched, and the
paper and the artifact report registry content only as aggregates, namely
counts and dates per platform.

\paragraph{Licenses}
The artifact ships scripts, manifests and derived data. OONI-derived data are
licensed CC BY-NC-SA 4.0, as their source requires. Sources without a license
are cited and not redistributed. Live sources such as the registrar's price
list can change, but the shipped derived data keep the paper's numbers.

\paragraph{Burn units}
Table~\ref{tab:burn} draws on GFWatch~\cite{gfwatch}, which lists only the
names that China blocked by DNS, from \mbox{2020-03-20} to \mbox{2024-09-06}. OONI
aggregates~\cite{ooni,ooni_data} cover China from July to September 2026
(\texttt{azurefd.net} from August 2025) and Iran and Russia from October 2024 to
September 2026. IRBlock~\cite{irblock} lists names found blocked by scans from
outside Iran between November 2024 and \mbox{2025-01-15}. The registry data are
seven snapshots~\cite{zapret_info} from \mbox{2021-12-31} to \mbox{2025-10-01}. A tenant name
is a registrable domain under a platform suffix~\cite{public_suffix_list}, or an
Azure Front Door endpoint name. In OONI, a hostname with at least
three non-failed measurements counts as tested, and as blocked if at least half
are anomalous
or confirmed. Single-label \texttt{workers.dev} account names do not resolve and are excluded. First dates of tenant blocks in GFWatch are upper bounds on their
onset, while suffix rules are dated by daily tests of the bare suffix, so
tenant-to-suffix gaps are lower bounds. A missing suffix rule shows
absence of evidence, not absence of blocking.

\paragraph{Snowflake}
Section~\ref{sec:snowflake} uses CollecTor's Snowflake broker
statistics~\cite{tor_collector}, Tor Metrics' bridge-user estimates and event
timeline~\cite{tor_metrics,tor_metrics_timeline}, 48 versions of Tor's
Connection Assist settings~\cite{tor_rdsys_admin} from \mbox{2022-02-25} to \mbox{2026-03-18}, and
Internet Outage Detection
and Analysis (IODA)
events~\cite{ioda_platform} (from \mbox{2022-01-26}) to exclude shutdown days. For the
Turkmenistan block and the end of Fastly's fronting, both open-ended, each change compares the median over the initial week with the median over the 14 valid days before. For Turkmenistan, user counts are the
lower bound of Tor Metrics' estimate.
User counts measure usage, not availability, and polls are not users. Derived tables omit the client credentials of Connection Assist.

\paragraph{Parameters and minting}
Section~\ref{sec:rotation} draws on a table of 110 values from 18 papers and
one field report, each with a verbatim quote found at its cited location. No
source measures $\lamburn$. Minting facts come from one registrar's price list
on one day~\cite{porkbun_pricing}, three free certificate authorities,
hosting platforms' documentation, \texttt{crt.sh} lookups of platform wildcard
certificates, ICANN's
zone-file access terms~\cite{icann_czds}, a CT honeypot study~\cite{ct_rise},
and the CT protocol, policies and log
list~\cite{rfc6962,chrome_ct_policy,apple_ct_policy,chrome_ct_log_list}. No
source shows whether any censor watches CT logs or zone files.

\section{Simulator, Solver and Reproducibility}\label{app:sim}

Our artifact is openly available at:
\begin{center}
{\urlstyle{tt}\artifacturl}
\end{center}

\paragraph{Event-driven simulator}
Our simulator draws sample paths of the model with the next-event method and
tracks the number $c$ of unblocked endpoints and the number $K$ of live names.
Since the two layers have independent clocks, we simulate them separately. The
deployment is down on the union of their down intervals. On a sample path, its
interval availability is the fraction of window starts $t_0$ for which
$[t_0,t_0+T]$ meets no down interval $[s_j,e_j]$. This fraction equals one
minus the share of window starts in $\bigcup_j(s_j-T,e_j)$, so we compute it
exactly, without discretization. A run lasts $1.2\times10^4$
to~$4\times10^4$ mint intervals. After discarding the first 10\% of every run
as warm-up, we average each point over 12 to 24 independent seeds. We compare
the simulator with the exact solver in Appendix~\ref{app:exact}.

\paragraph{Reproducing the results}
Both the simulator and the solver are in the artifact, along with every result
file, the two figures and the data pipelines of Appendix~\ref{app:data}. In
\texttt{simulations/\allowbreak code}, start with the solver \texttt{exact.py},
which runs its self-test. On a 20-core machine, \texttt{exact\_experiments.py}
recomputes the exact results in about five minutes. The nine simulation
experiments and the closed-form grid of Appendix~\ref{app:exact} come from
\texttt{run\_experiments.py}, and rerunning them takes about an hour on one
core. The artifact's README lists the library builds behind our results and
the checks we ran. With those builds and the fixed simulation seeds, reruns
reproduce the stored results to within $3\times10^{-14}$. Most of the values
that no result file stores, such as the wall-clock thresholds of
Section~\ref{sec:rotation}, appear in the output of \texttt{extra\_values.py}.
To redraw Figures~\ref{fig:frontier} and~\ref{fig:phase} and check them
against the solver, run
\texttt{paper/\allowbreak figures/\allowbreak make\_figures.py} from the
artifact's root. The independent implementation of Appendix~\ref{app:exact}
and the reimplementation of Appendix~\ref{app:data} are separate checks
outside the artifact.

\end{document}